\documentclass[oneside,onecolumn]{article}
\usepackage{microtype} 

\usepackage[english]{babel} 

\newcommand{\indep}{\perp \!\!\! \perp}
\newcommand{\uni}[2]{\mathrm{Uni}({#1{:}#2})}
\newcommand{\rd}[2]{\mathrm{Red}({#1{:}#2})}
\newcommand{\syn}[2]{\mathrm{Syn}({#1{:}#2})}
\newcommand{\SP}{\ensuremath{\mut{Z}{\hat{Y}}}}
\newcommand{\EO}{\ensuremath{\mut{Z}{\hat{Y} | Y}}}
\newcommand{\PP}{\ensuremath{\mut{Z}{Y | \hat{Y}}}}

\newcommand{\uniZ}{\ensuremath{\uni{Z}{\hat{Y} | Y}}}
\newcommand{\redZ}{\ensuremath{\rd{Z}{\hat{Y}, Y}}}
\newcommand{\synZ}{\ensuremath{\syn{Z}{\hat{Y}, Y}}}
\newcommand{\uniZY}{\ensuremath{\uni{Z}{Y | \hat{Y}}}}

\let\svthefootnote\thefootnote
\newcommand\freefootnote[1]{%
  \let\thefootnote\relax%
  \footnotetext{#1}%
  \let\thefootnote\svthefootnote%
}

\newcommand{\mut}[2]{\mathrm{I}({#1;#2})}
\newcommand{\mutd}[3]{\mathrm{I}_{#1}({#2;#3})}
\newcommand{\iid}[0]{i.i.d.}
\usepackage{amsthm}
\usepackage{thmtools}
\usepackage{thm-restate}

\usepackage{amsmath}

\newcommand{\Yh}{\hat{Y}}

\newtheorem{proposition}{Proposition}

\newtheorem{definition}{Definition}

\newtheorem{example}{Example}

\usepackage[disable,textsize=tiny]{todonotes}
\usepackage[textsize=tiny]{todonotes}

\usepackage{multirow}
\usepackage{booktabs}

\usepackage[utf8]{inputenc} 
\usepackage[T1]{fontenc}
\usepackage{url}
\usepackage{ifthen}
\usepackage{cite}
\usepackage{nicefrac}
\usepackage[hmarginratio=1:1,top=32mm,columnsep=20pt]{geometry} 

\begin{document}
\title{A Unified View of Group Fairness Tradeoffs \\ Using Partial Information Decomposition}

\author{Faisal Hamman and Sanghamitra Dutta \\
\normalsize University of Maryland College Park
}

\date{}
\maketitle


\begin{abstract}
This paper introduces a novel information-theoretic perspective on the relationship between prominent group fairness notions in machine learning, namely statistical parity, equalized odds, and predictive parity. It is well known that simultaneous satisfiability of these three fairness notions is usually impossible, motivating practitioners to resort to approximate fairness solutions rather than stringent satisfiability of these definitions. However, a comprehensive analysis of their interrelations, particularly when they are not exactly satisfied, remains largely unexplored. Our main contribution lies in elucidating an exact relationship between these three measures of (un)fairness by leveraging a body of work in information theory called partial information decomposition (PID).  In this work, we leverage PID to identify the granular regions where these three measures of (un)fairness overlap and where they disagree with each other leading to potential tradeoffs. We also include numerical simulations to complement our results.
\end{abstract}

\section{Introduction}
\freefootnote{Presented at the \textit{IEEE International Symposium on Information Theory} (ISIT 2024) in Athens, Greece.}
\freefootnote{The authors are with the Department of Electrical and Computer Engineering, University of Maryland College Park, United States 20742. Author contacts: \textit{fhamman@umd.edu}, \textit{sanghamd@umd.edu}}
\freefootnote{This work was supported in part by NSF CAREER Award 2340006.}

The increasing adoption of machine learning (ML) in high-stakes applications such as employment, finance, healthcare, etc, promises enhanced efficiency and improved decision-making.  However, this widespread reliance on ML systems has escalated concerns about the disparate impact~\cite{WhiteHouse2022AIBill,newFairSurvey,mehrabi2019survey,barocas-hardt-narayanan,varshney2021trustworthy,kamishima2011fairness,dutta2021fairness,hamman2023can,anthis2024causal} that these systems might cause on unprivileged \emph{groups} based on sensitive attributes such as gender, race, age, etc. Several anti-discrimination laws and ethical principles~\cite{WhiteHouse2022AIBill} are being actively put forth to ensure algorithmic fairness.

Existing literature has studied a plethora of definitions, metrics, and scholarly debates about algorithmic fairness~\cite{newFairSurvey,mehrabi2019survey}. Central to the debate of quantifying fairness at a group level are three popular definitions, namely, statistical parity, equalized odds, and predictive parity~\cite{newFairSurvey,hardt2016equality,mehrabi2019survey}. Due to the multitude of fairness definitions available, it is often unclear which measure of fairness is most appropriate to adopt in a given setting~\cite{washington2018argue}. Furthermore, it is also well-known that simultaneous satisfiability of these three fairness definitions is generally impossible~\cite{kleinberg2016inherent, chouldechova2017fair, Barocas2016BigDD}. 

Given such a fundamental impossibility,  practitioners often strive for approximate fairness solutions rather than stringent satisfiability of all these definitions. Such approximate fairness solutions consist of two pivotal aspects: (i) quantification
of (un)fairness (i.e., a gap from exact satisfiability); and (ii) development
of strategies to mitigate such unfairness in ML models. For instance, one may jointly minimize one or more measures of unfairness while training an ML model which has often led to empirical tradeoffs between accuracy and different measures of unfairness~\cite{dutta2020chernoff,kim2020fact}. 

Although previous studies have identified certain impossibilities among these fairness notions~\cite{chouldechova2017fair, kleinberg2016inherent, Barocas2016BigDD}, a detailed analysis focusing on \emph{the interrelationships among different measures of unfairness}, specifically explaining when they will be in agreement and when they will be in disagreement leading to potential tradeoffs has received limited attention.

Our research bridges this gap by leveraging Partial Information Decomposition (PID)~\cite{bertschinger_QUI}, a body of work in information theory, to elucidate the exact relationship between different measures of unfairness. In particular, we consider information-theoretic quantifications~\cite{ghassami2018fairness} of the respective gaps from statistical parity, equalized odds, and predictive parity as our measures of unfairness. Using PID, we demonstrate the exact relationship between these three measures of unfairness in Proposition~\ref{prop:decomposition}. We also refer to Fig.~\ref{fig:chainrule} and Fig.~\ref{fig:circle} for a pictorial illustration of the relationship between the measures of unfairness. 


PID enables us to provide a unified information-theoretic framework that is instrumental in establishing the fundamental limits and tradeoffs among these unfairness measures, particularly in the context of approximate fairness solutions when exact satisfiability of all three fairness definitions is not met. Furthermore, the impossibility among the three fairness definitions can also be derived from our result (see Theorem~\ref{thm:imp}). We also identify and delineate the regions of agreement and disagreement among these three measures of unfairness (see Section~\ref{sec:fairdecomp}), providing insights on when there will be a tradeoff and when there will be no tradeoff among the measures of unfairness.  We perform numerical simulations on the Adult dataset~\cite{Adult} to complement our theoretical results. Moreover, our work holds broader implications in fields such as algorithmic fairness auditing~\cite{hamman2023can}, where it can significantly contribute to the evaluation of fairness in ML models.





\begin{figure}
\centering
\includegraphics[width=0.7\textwidth]{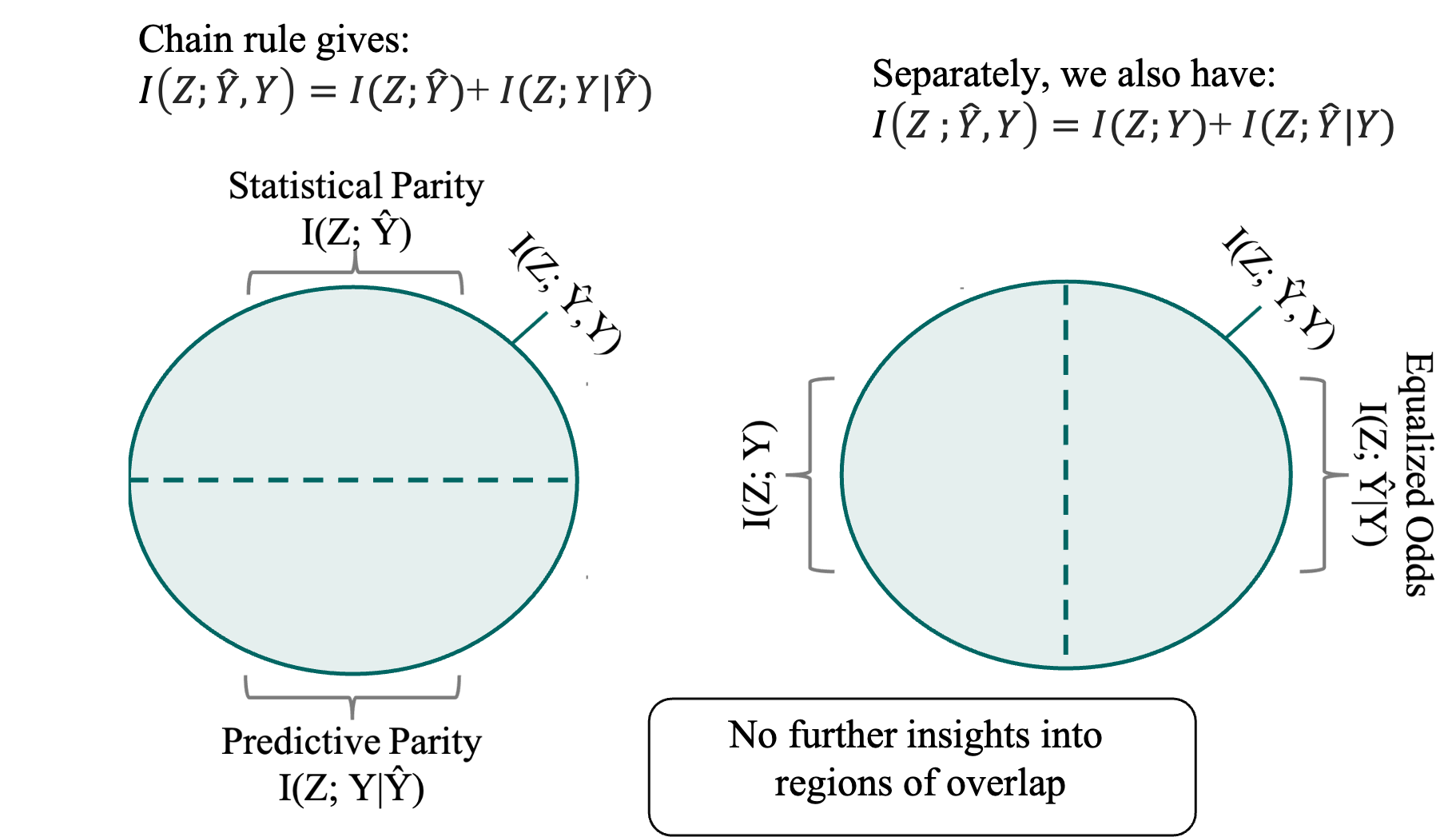}
  \caption{
Illustrates the decomposition of mutual information $\mut{Z}{\hat{Y},Y}$ using the chain rule. (\emph{left}) shows the decomposition into Statistical Parity and Predictive Parity. (\emph{right}) shows the decomposition into $\mut{Z}{Y}$ and Equalized Odds. No further insights into the overlapping regions of these measures, highlighting the need for measures to capture the nuanced interactions between fairness measures.  \label{fig:chainrule}}
\end{figure}

\subsection{Related Works}

Existing literature has introduced several foundational definitions and methods for fairness in machine learning (see some comprehensive overviews and surveys~\cite{barocas-hardt-narayanan,varshney2021trustworthy, mehrabi2019survey, newFairSurvey}). Another line of work focuses on exploring trade-offs between fairness and accuracy~\cite{wang2023aleatoric,dutta2020chernoff,zhao2022inherent,chen2018my,kim2020fact,long2023individual,zhong2024intrinsic,hamman2023demystifying,venkatesh2021can,sabato2020bounding,menon2017cost,liu2022accuracy,hertweck2022gradual,hsu2022pushing}. 

Early works on impossibility results among different fairness measures highlight the challenges of simultaneously satisfying multiple group fairness criteria~\cite{chouldechova2017fair, kleinberg2016inherent, NIPS2016_9d268236, pleiss2017fairness,Barocas2016BigDD}. Notably, \cite{bell2023possibility} challenges the practical implications of the impossibility theorem, showing that fairness across multiple criteria is more achievable than previously believed. Similarly,~\cite{hsu2022pushing} presents an integer-programming-based framework for optimizing post-processing methods to simultaneously satisfy multiple fairness criteria under small violations while maintaining a minimal reduction in model performance. However, the nuanced interrelationships among different
measures of unfairness, specifically explaining when they will be in agreement and when they will be in disagreement leading
to potential tradeoffs has received limited attention.

Information-theoretic measures have been used to quantify group fairness in existing literature \cite{kamishima2012fairness,calmon2017optimized,ghassami2018fairness,dutta2021fairness,dutta2020information,cho2020fair,baharlouei2019r,grari2019fairness,10.1109/ISIT45174.2021.9517723,galhotra2022causal,NEURIPS2022_fd5013ea,kairouz2022generating}. For instance, \cite{ghassami2018fairness,dutta2021fairness} have already used mutual information and conditional mutual information to quantify statistical parity, equalized odds, and predictive parity. Closely related to our work,~\cite{hertweck2022gradual} explores whether fairness measures can be gradually compatible using these information-theoretic quantifications and chain rule, also showing that some fairness criteria can be simultaneously improved through fairness-regularized predictors. Our work dives deeper into this nuanced tradeoff among the three group fairness measures leveraging a body of work in information theory called Partial Information Decomposition (PID) that goes beyond chain rule, delineating regions of agreement and disagreement.

PID is recently gaining traction across various ML applications~\cite{dutta2023review,dutta2020information,dutta2021fairness,hamman2023demystifying,tax2017partial,liang2023quantifying,mohamadi2023more,pakman2021estimating,halder2024quantifying,venkatesh2023capturing}.
It is particularly noteworthy in the realm of algorithmic fairness~\cite{dutta2023review,dutta2020information,dutta2021fairness,hamman2023demystifying}. Prior work~\cite{dutta2020information,dutta2021fairness} leverages PID to dissect total disparity in decision-making into exempt and non-exempt components depending on which features they came from. Another work~\cite{hamman2023demystifying} leverages PID to analyze the trade-offs between global and local fairness in a federated learning setting, identifying three sources of unfairness, and formulates a convex optimization problem to define the theoretical limits of accuracy and fairness trade-offs. We also refer to \cite{dutta2023review} for a survey of PID in fairness and explainability. Understanding tradeoffs and agreement disagreement between the three canonical group fairness measures, namely statistical parity, equalized odds, and predictive parity, using
PID has not been studied.
We aim to develop a unified information-theoretic framework that effectively formalizes the fundamental limits and trade-offs among these three unfairness measures.

\section{Preliminaries}\label{prel}

Let $X$ denote the input features, $Z$ denote the sensitive attribute, and $Y$ denote the true label. The sensitive attribute $Z$ is assumed to be binary with $1$ indicating the privileged group and $0$ indicating the unprivileged group. We also let $\hat{Y}$ represent the predictions of a model, i.e., $\hat{Y}= f_\theta(X)$ where the model is parameterized by $\theta$. Standard machine learning aims to minimize the empirical risk: 
\begin{equation}
    \min_\theta L(\theta) = \min_\theta \frac{1}{n} \sum_{i=1}^{n} l(f_\theta(x_i),y_i),
\end{equation} where $l(\cdot,\cdot)$ is a predefined loss function, $x_i$ is the input feature, $y_i \in \{0,1\}$ is the true label, and $n$ is the number of datapoints in the dataset.

\subsection{Background on Partial Information Decomposition}

\begin{figure}
  \centering
\includegraphics[width=0.3\textwidth]{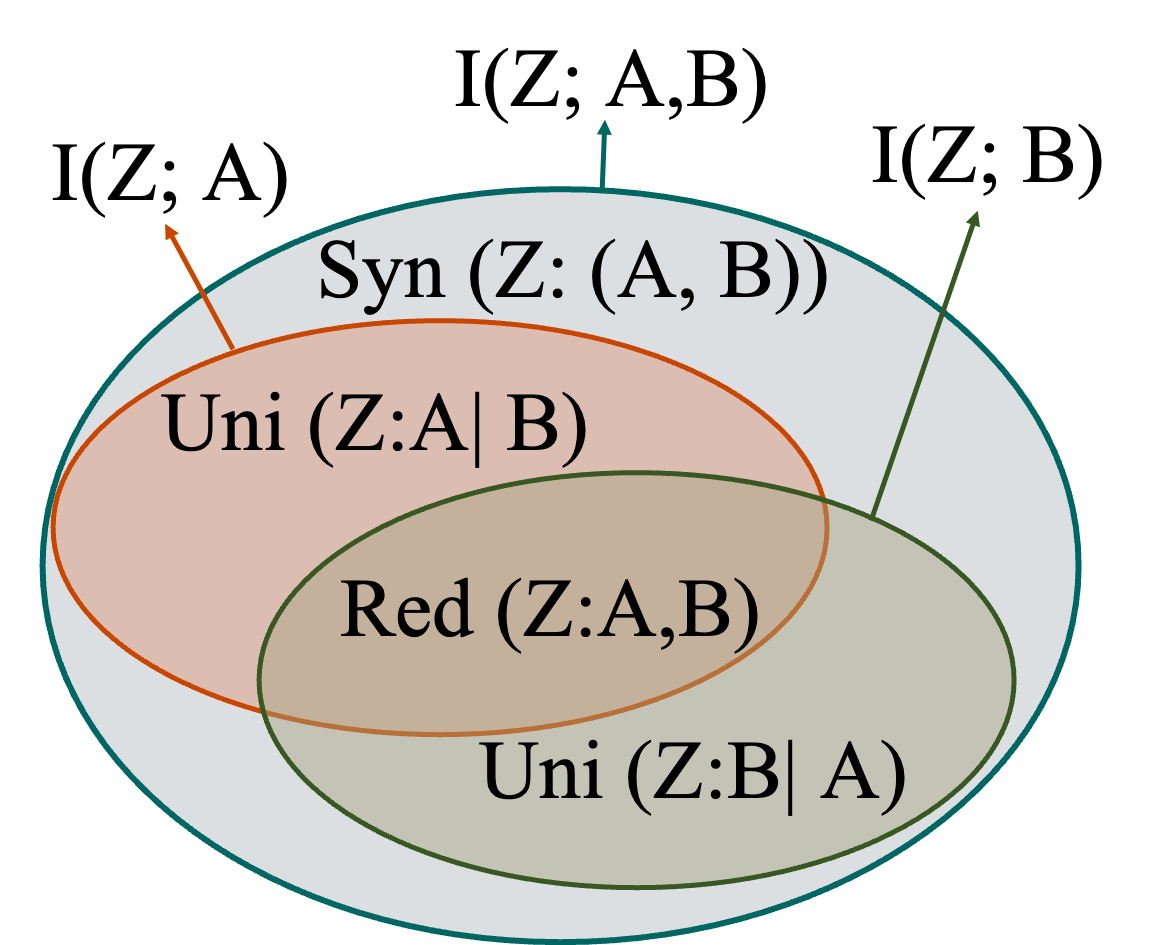}
  \caption{Venn diagram showing PID of $\mut{Z}{A,B}$.
  \label{fig:sim}}
\end{figure}

Partial Information Decomposition (PID)~\cite{bertschinger_QUI} 
decomposes the total mutual information about a random variable Z contained in the tuple $(A, B)$, i.e., $\mathrm{I}(Z; A, B)$  into four \emph{nonnegative} terms as follows (also see Fig.~\ref{fig:sim}):
\begin{align}
\mathrm{I}(Z; A, B)  = \uni{Z}{A|B}  +\uni{Z}{B|A}\label{eq:pid} \\+\rd{Z}{A,B} + \syn{Z}{A,B} \nonumber 
\end{align}
Here, $\uni{Z}{A|B}$ denotes the unique information about $Z$ that is present only in $A$ and not in $B$. E.g., \textit{shopping preferences} ($A$) may provide unique information about \textit{gender} ($Z$) that is not present in \textit{address} ($B$). $\text{Red}(Z{:} A, B)$ denotes the redundant information about $Z$ that is present in both $A$ and $B$. E.g., \textit{zipcode} ($A$) and \textit{county} ($B$) may provide redundant information about \textit{race} ($Z$). The term $\syn{Z}{A,B}$ denotes the synergistic information not present
in either $A$ or $B$ individually, but present jointly in $(A, B)$, e.g., each individual digit of the \textit{zipcode} may not have information about \textit{race} but together they provide significant information. Before formally defining these terms, we provide an example.

\noindent \textbf{Motivational Example.} Let $Z{=}(Z_1,Z_2,Z_3)$ with each $Z_i{\sim}$ \iid{} Bern(1/2). Let $A=(Z_1,Z_2,Z_3\oplus N)$, $B=(Z_2,N)$, and $N\sim $  Bern(1/2) which is independent of $Z$. Here,  $\mathrm{I}(Z;A,B)=3$ bits. The unique information about $Z$ that is contained only in $A$ and not in $B$ is effectively in $Z_1$, and is given by $\uni{Z}{A| B} = \mut{Z}{Z_1} = 1$ bit. The redundant information about $Z$ that is contained in both $A$ and $B$ is effectively in $Z_2$ and is given by $\rd{Z}{A, B}=\mathrm{I}(Z;Z_2)=1$ bit. Lastly, the synergistic information about $Z$ that is not contained in either $A$ or $B$ alone, but is contained in both of them together is effectively in the tuple $(Z_3\oplus N,N)$, and is given by $\syn{Z}{A,B} {=} \mut{Z}{(Z_3\oplus N,N)}=1 $ bit. This accounts for the $3$ bits in $\mut{Z}{A,B}$. 

We also note that defining any one of the PID terms suffices in obtaining the others. This is because of another relationship among the PID terms as follows~\cite{bertschinger_QUI}: $\mut{Z}{A}=\uni{Z}{A|B} + \rd{Z}{A, B}$.  Essentially $\rd{Z}{A, B}$ is viewed as the sub-volume between $\mut{Z}{A}$ and $\mut{Z}{B}$ (see Fig.~\ref{fig:sim}). Hence,  $\rd{Z}{A, B} = \mut{Z}{A}- \uni{Z}{A|B}$. Lastly, $\syn{Z}{A,B} =  \mathrm{I}(Z; A, B) - \uni{Z}{A|B} -\uni{Z}{B|A}-\rd{Z}{A, B}$ (can be obtained from \eqref{eq:pid} once both unique and redundant information has been defined). The main results of our paper hold regardless of the specific 
 definition of a given PID term. However, our experiments are based on the precise definition of $\uni{Z}{A|B}$ from \cite{bertschinger_QUI}.


\begin{definition}[Unique Information~\cite{bertschinger_QUI}]\label{def:bert_def} Let $\Delta$ be the set of all joint distributions on $(Z, A, B)$ and $\Delta_p$ be the set of joint distributions with the same marginals on $(Z, A)$ and $(Z, B)$ as the true distribution, i.e., $\Delta_p=\{Q {\in} \Delta$ : $\Pr_Q(Z=z, A=a) =\Pr(Z=z, A=a)$ and $\Pr_Q(Z=z,B=b)=\operatorname{Pr}(Z=z, B=b)\}$. Then, 
\begin{equation}\uni{Z}{A|B}=\min _{Q \in \Delta_p} \mathrm{I}_Q(Z ; A | B),
\end{equation}
where $\mathrm{I}_Q(Z ; A | B)$ is the conditional mutual information when $(Z, A, B)$ have joint distribution $Q$ and $\Pr_Q(\cdot)$ denotes the probability under $Q$. 
\end{definition}

\noindent \textbf{Operational meaning of Unique Information from Blackwell sufficiency:} Unique information is closely tethered to Blackwell Sufficiency~\cite{blackwell1953equivalent} in statistical decision theory. The concept of Blackwell sufficiency~\cite{blackwell1953equivalent} from statistical decision theory helps characterize if a random variable $A$ is more informative than $B$ about $Z$ (also relates to stochastic degradation of channels~\cite{banerjee2018unique,venkatesh2023capturing}). A channel $P_{B|Z}$ is Blackwell sufficient with respect to another channel $P_{A|Z}$ (also denoted as $B \geq_{Z} A$) if there exists a stochastic transformation $P_{A'|B}$ such that the effective channel from $Z$ to $A'$ is equivalent to the original channel from $Z$ to $A$ (see Fig.~\ref{fig:blackwell}). The unique information $\uni{Z}{A|B}$ is $0$ if and only if $P_{B|Z}$ is Blackwell sufficient with respect to $P_{A|Z}$~\cite{bertschinger_QUI,venkatesh2023capturing,banerjee2018unique}. Otherwise, $\uni{Z}{A|B}>0$, and it is viewed as a departure from Blackwell sufficiency, i.e., \emph{there exists a scenario where $A$ gives something unique about $Z$ that you can never get after degrading to $B$.} 

\begin{figure}[!htbp]
\centering
\includegraphics[width=4cm]{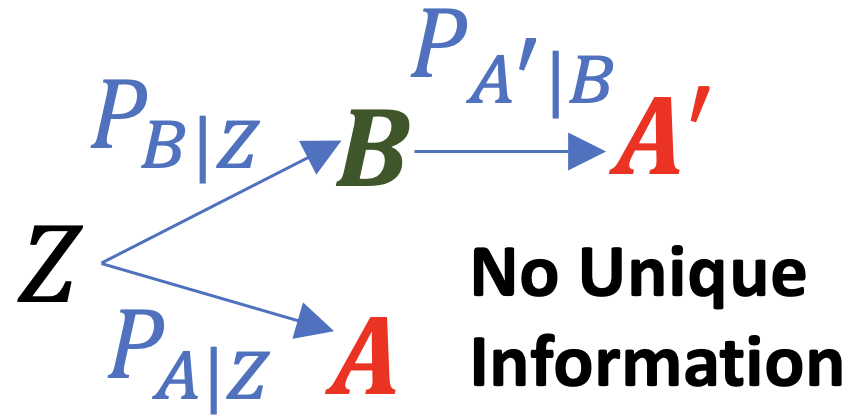}
\caption{Blackwell sufficiency of channel $P_{B|Z}$ with respect to $P_{A|Z}$ means $A$ has no unique information about $Z$ that is not in $B$.\label{fig:blackwell}}
\end{figure}

\section{Decomposition of the Measures of Unfairness}\label{sec:fairdecomp}

We first introduce the information-theoretic quantification corresponding to the three definitions of fairness, namely, statistical parity, equalized odds, and predictive parity. 
 Statistical parity (\textit{independence}), requires the model prediction $\hat{Y}$ to be statistically independent of the sensitive attribute $Z$. Several measures have been proposed to quantify the gap from statistical parity \cite{shi2021survey,mehrabi2019survey} (essentially dependence between $\hat{Y}$ and $Z$). 
In this work, we use the information-theoretic quantification of the statistical parity gap as defined next.

\begin{definition}[Statistical Parity Gap]
The statistical parity gap of a model \( f_\theta \) with respect to \( Z \) is defined as \( I(Z; \hat{Y}) \), the mutual information between \( Z \) and \( \hat{Y} \) (where \( \hat{Y}=f_\theta(X) \)).
\end{definition}

The concept of statistical parity has often been criticized for not considering the true labels. 
A perfect predictor \( \hat{Y} = Y \) might not satisfy this criterion if $Y$ is correlated to the sensitive attribute $Z$. Hence, the concept of equalized odds emerges as an alternative definition of fairness~\cite{hardt2016equality}. Equalized odds (\emph{separation}) require the model's predictions $\hat{Y}$ to be independent of the sensitive attribute $Z$, conditioned on the true label $Y$, i.e.,  \(Z \indep \hat{Y}|Y \).

\begin{definition}[Equalized Odds Gap] \label{def:EO}
The equalized odds gap of a model $f_\theta$ with respect to $Z$ is defined as $\mut{Z}{\hat{Y}|Y}$, the conditional mutual information between $Z$ and $\hat{Y}$ given $Y$.
\end{definition}

Yet another vital fairness measure is predictive parity (\emph{sufficiency}), which focuses on error parity among individuals 
given the same prediction \cite{newFairSurvey}. Predictive parity requires the sensitive attribute $Z$ to be independent of the true label $Y$ conditioned on the model prediction $\hat{Y}$, i.e., \( Z \indep Y|\hat{Y} \). 
\begin{definition}[Predictive Parity Gap] \label{def:PP}
The predictive parity gap of a model $f_\theta$ with respect to $Z$ is defined as $\mut{Z}{Y|\hat{Y}}$, the conditional mutual information between $Z$ and $Y$ given $\hat{Y}$.
\end{definition}


We leverage PID to derive exact relationships among the three measures of unfairness. We decompose the statistical parity gap $I(Z; \hat{Y})$, equalized odds gap $I(Z; \hat{Y}|Y)$, and predictive parity gap $I(Z; Y|\hat{Y})$ into \emph{nonnegative} overlapping terms. The significance of this decomposition is that it highlights regions where these measures are in agreement and disagreement. Fig.~\ref{fig:chainrule} and Fig.~\ref{fig:circle}  provides a pictorial illustration of the overlaps between these three measures of unfairness.

\begin{proposition}\label{prop:decomposition} The statistical parity gap $\mut{Z}{\hat{Y}}$, equalized odds gap $\mut{Z}{\hat{Y}|Y}$, and predictive parity gap $\mut{Z}{Y|\hat{Y}}$ can be decomposed into nonnegative terms as follows:
\begin{align}
&\mut{Z}{\hat{Y}}=\uni{Z}{\hat{Y}|Y} + \rd{Z}{\hat{Y}, Y}. \\
&\mut{Z}{\hat{Y}|Y}=\uni{Z}{\hat{Y}|Y} + \syn{Z}{\hat{Y},Y}. \\
&\mut{Z}{Y|\hat{Y}}=\uni{Z}{Y|\hat{Y}} + \syn{Z}{\hat{Y},Y}. 
\end{align} 
\end{proposition}

The term $\uni{Z}{\hat{Y}|Y}$  quantifies the unique information about the sensitive attribute $Z$ in the model prediction $\hat{Y}$ that is not there in the true label $Y$.  $\uni{Z}{\hat{Y}|Y}$  is the common region between the statistical parity gap and the equalized odds gap, highlighting the region where they overlap. The term $\rd{Z}{\hat{Y}, Y}$ quantifies the information about sensitive attribute $Z$ that is common between prediction $\hat{Y}$ and true label $Y$. $\rd{Z}{\hat{Y}, Y}$ contributes only to the statistical parity gap $\mut{Z}{\hat{Y}}$ and not to any other measure of unfairness. The term $\syn{Z}{\hat{Y},Y}$ represents the synergistic information about sensitive attribute $Z$ that is \emph{not} present in either $\hat{Y}$ or $Y$ individually but is present jointly in $(\hat{Y}, S)$. $\syn{Z}{\hat{Y},Y}$ is the common region between equalized odds gap and predictive parity gap, highlighting their region of agreement. The unique information $\uni{Z}{Y|\hat{Y}}$  contributes exclusively to the predictive parity gap \( I(Z; Y|\hat{Y}) \). This decomposition delineates the distinct regions where these unfairness measures overlap and diverge, offering a nuanced perspective on the interplay in machine learning models.


To better illustrate this decomposition, we now provide examples to understand each of these regions separately. Consider a hiring scenario featuring binary sensitive attributes and true labels i.e., $\hat{Y}, Z, Y \in \{0,1\}$ with $Z{\sim}$ Bern(1/2).


\begin{figure}
  \centering
\includegraphics[width=0.4\textwidth]{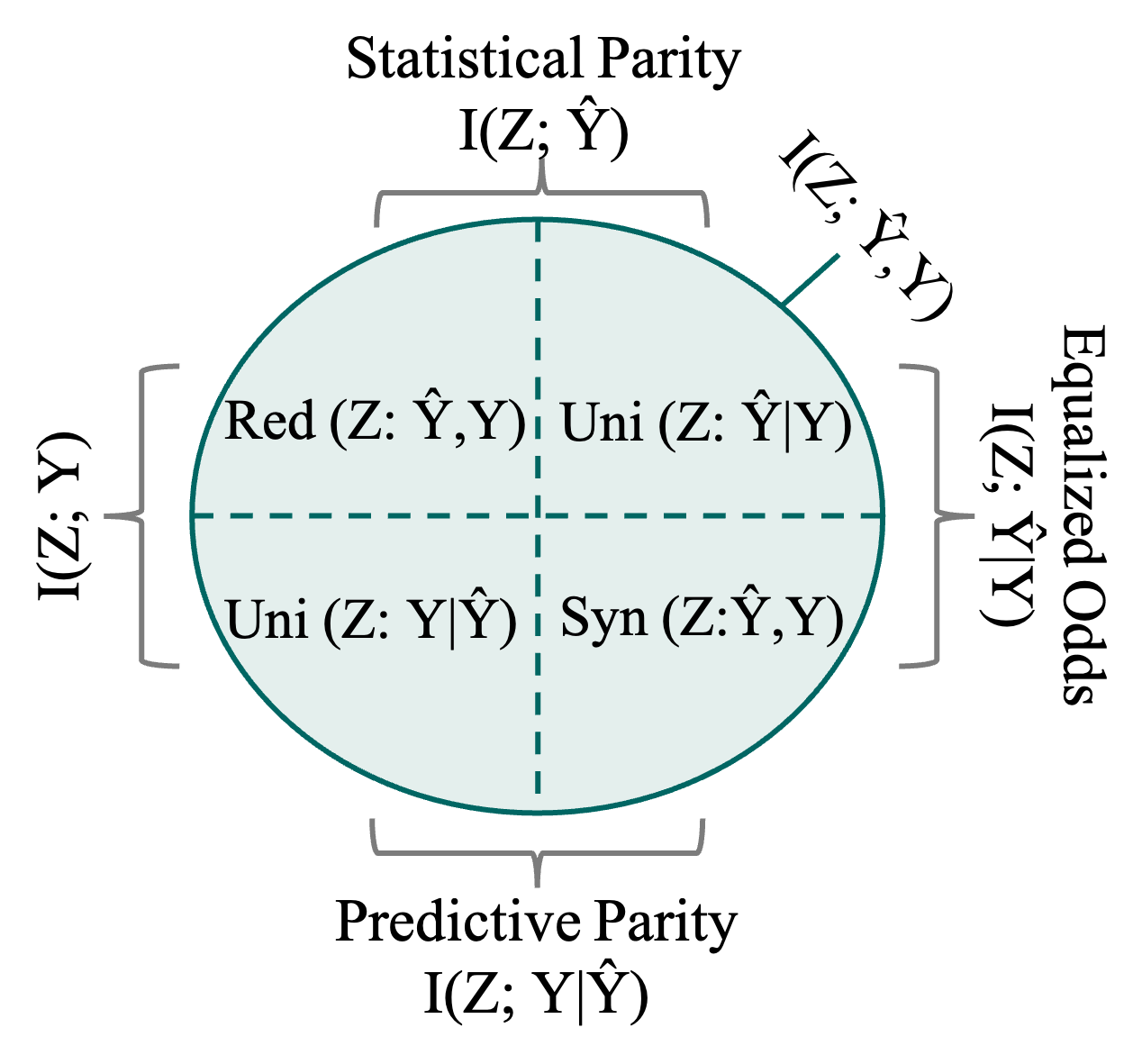}
  \caption{Venn diagram showing the exact relationship between the various unfairness measures using PID: A critical observation is that all four PID terms are nonnegative. This enables us to derive several fundamental limits and tradeoffs among the unfairness measures, providing a nuanced understanding of when they agree and disagree.
  \label{fig:circle}}
  \vspace{-5pt}
\end{figure}

\begin{example}[Pure Uniqueness to Model Prediction]
Let $\hat{Y}=Z$ and $Z \indep Y$ (an equal base rate for privileged and unprivileged groups). Suppose, the model only approves privileged candidates ($Z=1$) but rejects the unprivileged ($Z=0$).  This model violates both statistical parity and equalized odds, i.e., $\mut{Z}{\hat{Y}}=\mut{Z}{\hat{Y}|Y}=1$. This model satisfies predictive parity criterion, $\mut{Z}{Y|\hat{Y}}=0$. This is a case of purely unique information in the model prediction that is not in the true label since all the information about $Z$ is derived exclusively from the model predictions; the true label $Y$ does not correlate with $Z$.  Here, $\uni{Z}{\hat{Y}|Y}=1$, $\rd{Z}{\hat{Y},Y}=0$, $\syn{Z}{\hat{Y},Y}=0$, and $\uni{Z}{Y|\hat{Y}}=0$.
\end{example}

\begin{example}[Pure Redundancy]\label{exp:pureRED}
Let $\hat{Y}=Y$ and $Y=Z$ with probability $0.9$. There is a correlation between the true label $Y$ and protected attribute $Z$, but this model has perfect accuracy. Such a model satisfies equalized odds and predictive parity criterion, i.e., $\mut{Z}{\hat{Y}|Y}=\mut{Z}{Y|\hat{Y}}=0$. However, the model fails to satisfy statistical parity since $\mut{Z}{\hat{Y}}=0.53.$  This is a case of purely redundant information since the information about $Z$ is entirely common between both $\hat{Y}$ and $Y$. Here, $\uni{Z}{\hat{Y}|Y}=0$, $\rd{Z}{\hat{Y},Y}=0.53$,  $\syn{Z}{\hat{Y},Y}=0$, and $\uni{Z}{Y|\hat{Y}}=0$.

\end{example}

\begin{example}[Pure Synergy]\label{exp:syn}
Let $\hat{Y}=Z$ XNOR $Y$ and $Z \indep Y$. The model approves candidates from the privileged group ($Z=1$) with true label $Y=1$, and also from the unprivileged group ($Z=0$) with $Y=0$. On the other hand, it rejects candidates from the unprivileged group ($Z=0$) with true label $Y=1$, and the privileged group ($Z=1$) with true label $Y=0$. Such a model violates equalized odds (and predictive parity) as it singularly prefers one group within each true label class. Thus, $\mut{Z}{\hat{Y}|Y}=1$, and $\mut{Z}{Y|\hat{Y}}=1$. However, it achieves statistical parity since it maintains an equal approval rate for both privileged and unprivileged groups with $\mut{Z}{\hat{Y}}=0$.  This is a case of synergistic information about $Z$ that is not observable in either $\hat{Y}$ or $Y$ individually but is present jointly in $Y,\hat{Y}$. Here, $\uni{Z}{\hat{Y}|Y}=0$, $\rd{Z}{\hat{Y},Y}=0$,  $\syn{Z}{\hat{Y},Y}=1$, and $\uni{Z}{Y|\hat{Y}}=0$.
\end{example}

\begin{example}[Pure Uniqueness to True Label]
Let $Y=Z$ with probability $0.9$ and $Z \indep \hat{Y}$. The true label $Y$ is highly correlated to sensitive attribute $Z$, but the model prediction $\hat{Y}$ is independent of sensitive attribute $Z$.  This model violates predictive parity ($\mut{Z}{Y|\hat{Y}}=0.53$) but satisfies statistical parity and equalized odds ($\mut{Z}{\hat{Y}} =\EO=0$). This is a case of unique information about sensitive attributes in the true label that is not in the model prediction. Here, $\uni{Z}{\hat{Y}|Y}=0$, $\rd{Z}{\hat{Y},Y}=0$, $\syn{Z}{\hat{Y},Y}=0$, and $\uni{Z}{Y|\hat{Y}}=0.53$.
\end{example}

These examples demonstrate scenarios of pure uniqueness, redundancy, and synergy to help us understand the decomposition. PID serves as a tool to highlight regions of agreement and disagreement between these fairness definitions.  In contrast, traditional fairness metrics lack the granularity to capture these nuanced interactions, making PID an essential asset for a more comprehensive understanding and mitigation of disparities. \\

We can go beyond the impossibility between the three fairness definitions and further analyze their interrelationships.

\begin{restatable}[Revisiting Impossibility]{theorem}{revisitingimposibility}\label{thm:imp}
If $\mut{Z}{\Yh,Y} {>} 0$, at least one of the  PID terms, namely, $\uniZ$, $\redZ$, $\synZ$, or $\uniZY$ will be nonnegative. Hence, at least one of the fairness measures, namely, the Statistical Parity Gap ($\SP$), Equalized Odds Gap ($\EO$), or Predictive Parity Gap ($\PP$) will be nonzero. Conversely, all these unfairness measures will be zero if and only if $\mut{Z}{\Yh,Y} = 0$.
\end{restatable}

\noindent \textit{Proof Sketch:} 
The proof relies on the nonnegativity of each of the PID terms (also recall Fig.~\ref{fig:circle}). PID of $\mut{Z}{\Yh,Y}$ is expressed as $\mut{Z}{\Yh,Y} = \uniZ + \uniZY + \redZ + \synZ$.
Since each component in this decomposition is nonnegative, the presence of mutual information ($\mut{Z}{\Yh,Y} > 0$) implies that at least one of these terms is nonzero. According to Proposition~\ref{prop:decomposition}, each of these PID terms influences at least one unfairness measure. Therefore, the nonnegativity of any one of these terms results in at least one of the unfairness measures being nonzero. \qed

This is a general result from which one can also derive the impossibility of the three fairness definitions under specific conditions. Our next result examines the unfairness measures only when $\mut{Z}{Y} >0$. It is important to note that $\mut{Z}{Y}$ is an inherent characteristic of the dataset alone and hence it is independent of the model predictions.
\begin{restatable}[Dataset Dependent Relationships]{theorem}{DatasetDependentRelationships}
If $\mut{Z}{Y} {>} 0$, either the Statistical Parity Gap $\SP$ or the Predictive Parity Gap $\PP$ must be greater than zero.
\end{restatable}

\noindent \textit{Proof Sketch:} 
The proof relies on demonstrating that the mutual information between $Z$ and $Y$ can be expressed as:
    \begin{equation}
       \mut{Z}{Y} = \uni{Z}{Y|\Yh}+\rd{Z}{Y,\Yh}. 
    \end{equation}     
Though, the PID terms $\uni{Z}{Y|\Yh}$ and $\rd{Z}{Y,\Yh}$ may vary based on the model chosen, their sum remains constant, reflecting the fixed nature of the mutual information between $Z$ and $Y$ in the dataset. Notably, $\uni{Z}{Y|\Yh}$ contributes to the predictive parity gap, and $\rd{Z}{Y,\Yh}$ contributes to the statistical parity gap (recall Fig.~\ref{fig:circle}). \qed

\begin{figure}
\centering
\includegraphics[width=0.7\textwidth]{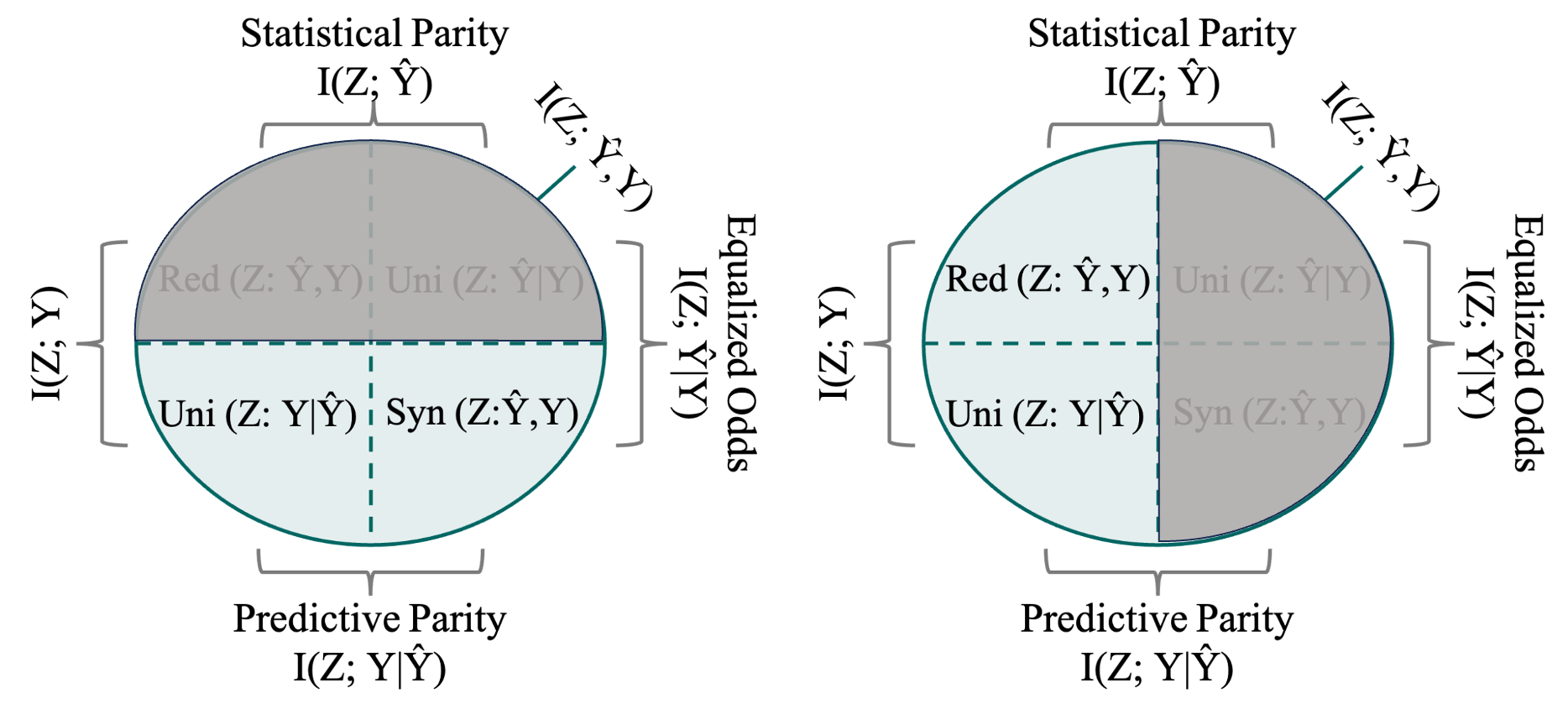}
  \caption{
(\emph{left}) Illustrates Theorem \ref{thm:zero_SP}, showing that when Statistical Parity is satisfied, the Predictive Parity gap is greater than or equal to the Equalized Odds gap, and if \(\mut{Z}{Y} {=} 0\), then \(\PP {=} \EO\).  (\emph{right}) visualizes Theorem \ref{thm:EOtradeoff} illustrating that when Equalized Odds is satisfied and \(\mut{Z}{Y} {>} 0\), there is an inverse relationship (tradeoff) between Statistical Parity and Predictive Parity (\(\SP {=} \mut{Z}{Y} {-} \PP\)) since \(\mut{Z}{Y}\) is fixed.
  \label{fig:theorem_viz}}
\end{figure}

\section{Tradeoffs Between Unfairness Measures}\label{sec:tradeoff}
In this section, we delineate the fundamental limits and tradeoffs between various unfairness measures. 
Our findings underscore the intricate and sometimes conflicting nature of different fairness objectives in algorithmic decision-making. 
Examining fairness through the lens of PID uncovers the nuanced interplay between different unfairness measures.

First, we will explore scenarios where models are trained with a focus on achieving any one specific fairness criterion and analyze its implications on the other two fairness notions. This applies to models that have been trained to achieve fairness either through in-processing techniques, such as adding fairness regularizers to the loss function, or through post-processing methods that adjust model outputs after training.


\begin{restatable}{theorem}{tradeoffone}\label{thm:zero_SP}
If Statistical Parity is satisfied, i.e., \(\SP = 0\), then the Predictive Parity Gap is greater than the Equalized Odds Gap, i.e., \(\PP \geq \EO\). Additionally, if the dataset is such that \(\mut{Z}{Y} = 0\), then Predictive Parity and Equalized Odds are equivalent, i.e., \(\PP = \EO\).
\end{restatable}

\noindent \textit{Proof Sketch:} We refer to Fig.~\ref{fig:circle} for an intuitive understanding of the proof. Given that Statistical Parity is zero, we have \(\SP = \uniZ + \redZ = 0\). Since all PID terms are non-negative, it follows that individually $\uniZ=0$ and $\redZ = 0$. Consequently, the Equalized Odds gap simplifies to \(\EO = \uniZ + \synZ = \synZ\). On the other hand, the Predictive Parity gap is \(\PP = \uniZY + \synZ\). Since all PID terms are nonnegative, it follows that \(\PP \geq \EO\).  

Furthermore, when \(\mut{Z}{Y} = 0\), it results in \(\uniZY + \redZ = 0\), leading to each of those individual terms being zero, i.e., $\uniZY=0$ and $\redZ=0$. Therefore, \(\PP =  \synZ\ = \EO\) (see Fig.~\ref{fig:theorem_viz} for an illustration). \qed


Similar to Theorem \ref{thm:zero_SP}, one can also derive the relationship between the statistical parity gap and equalized odds gap when predictive parity is satisfied. 
\begin{restatable}{theorem}{tradeofftwo}\label{thm:zero_PP}
If Predictive Parity is satisfied, i.e., \(\PP {=} 0\), then the Statistical Parity Gap is greater than the Equalized Odds Gap, i.e., \(\SP \geq \EO\). Additionally, if the dataset is such that \(\mut{Z}{Y} = 0\), then Statistical Parity and Equalized Odds are equal, i.e., \(\PP = \EO\).
\end{restatable}
\noindent Theorem \ref{thm:zero_SP} and \ref{thm:zero_PP} demonstrate scenarios where one unfairness measure dominates another and are in agreement, now we provide a third scenario where two measures of unfairness will be in disagreement.

\begin{restatable}{theorem}{tradeoffthree}\label{thm:EOtradeoff}
If Equalized Odds is satisfied, i.e., \(\EO {=} 0\) and $\mut{Z}{Y} >0$, an inverse relationship (tradeoff) exists between Statistical Parity and Predictive Parity, i.e., $\SP=\mut{Z}{Y}-\PP$. Thus, increasing one leads to a decrease in the other, and vice versa.
\end{restatable}
\noindent \textit{Proof Sketch}
Given that Equalized Odds is met, we have \(\EO = \uniZ + \synZ = 0\). Consequently, from nonnegativity, both the terms $\uniZ$ and $\synZ$ are 0. Statistical Parity gap simplifies to \(\SP = \uniZ + \redZ = \redZ\), and the Predictive Parity gap is expressed as \(\PP = \uniZY + \synZ = \uniZY\). Hence, $\mut{Z}{Y} = \uniZY + \redZ = \SP + \PP$. Since, $\mut{Z}{Y}$ is fixed for a dataset, an increase in the statistical parity gap leads to a decrease in the predictive parity gap, and vice versa (see Fig.~\ref{fig:theorem_viz} for an illustration). \qed

\section{Experimental Demonstrations}

\begin{table*}[!htbp]
\centering
\renewcommand{\arraystretch}{1.1}
\setlength{\tabcolsep}{18pt}
\caption{Results of Regularizers on Different Measures of Unfairness}
\label{tab}
\begin{tabular}{@{}lcccc@{}}
\toprule
\multirow{3}{*}{Regularizers} & \multicolumn{1}{l|}{}    & \multicolumn{2}{l|}{Equalized Odds $\EO$}                             &     \\ \cline{2-5} 
                            & \multicolumn{2}{l|}{Statistical Parity $\SP$}                             & \multicolumn{2}{l}{Predictive Parity $\PP$}        \\ \cline{2-5} 
                            & \multicolumn{1}{l|}{$\redZ$} & \multicolumn{1}{l|}{$\uniZ$} & \multicolumn{1}{l|}{$\synZ$} & $\uniZY$ \\ \hline

SP                    & 0.012                  & 0.000                   & 0.001                  & 0.024                       \\ 
PP                    & 0.026                  & 0.007                   & 0.008                  & 0.011                        \\ 
EO                    & 0.011                  & 0.000                   & 0.001                  & 0.026                       \\ 

EO, PP                & 0.000                  & 0.000                   & 0.000                  & 0.037                       \\ 
SP, PP                & 0.000                  & 0.000                   & 0.000                  & 0.037                       \\ 
SP, EO                & 0.008                  & 0.000                   & 0.000                  & 0.028                       \\ 
SP, EO, PP            & 0.000                  & 0.000                   & 0.000                  & 0.037                       \\ 
\bottomrule
\end{tabular}
\end{table*}

In this section, we provide an experimental demonstration on the Adult dataset~\cite{Adult} to validate our theoretical findings. The classification task for this dataset involves predicting whether an individual's income exceeds $50$K per year, using features such as occupation, education, etc. We use \textit{gender}  as a sensitive attribute.

We train a neural network consisting of a sequence of layers: the input layer is followed by three hidden layers, each with $32$ units and ReLU activation, and concludes with a single output layer using a sigmoid activation function. Training is conducted using a batch size of $512$, and the Adam optimizer with a learning rate of $0.01$. We apply various fairness regularizers and measure the unfairness as well as their decomposition (results are summarized in Table.\ref{tab}).
We use the \textit{dit} package~\cite{dit} for PID computation and \textit{FairTorch} \cite{fairtorch} for fairness regularizer implementation.

 A key observation in our analysis is that  \(\mut{Z}{Y}\) consistently measures 0.037 using the Adult dataset.  This mass does not decrease across various models since it only depends on the dataset. The PID terms in $\mut{Z}{Y}$, i.e., $\uniZY$ and $\redZ$ contribute to either predictive parity or statistical parity gap. 
 When statistical parity is achieved (scenario with SP regularizer), the predictive parity gap is greater than the equalized odds gap. Also due to the impossibility of attaining zero unfairness with all the measures (see scenario with SP, EO, and PP regularizers), the mass typically moves to $\uniZY$, contributing to the predictive parity. The experiments confirm the theoretical insights by illustrating the inherent trade-offs between fairness measures. Specifically, even with fairness regularizers, achieving zero unfairness across all metrics is impossible due to the fixed information content in the dataset, i.e., $\mut{Y}{Z}$. The movement of unfairness mass between statistical parity and predictive parity highlights the difficulty of balancing fairness across multiple criteria, reinforcing the necessity of carefully considering these trade-offs in real-world applications.

\section{Conclusion}
By introducing this unifying framework, we provide a tool for gaining a more nuanced understanding of the interplay between different unfairness measures, thereby improving the decision-making and deployment of fair machine learning systems. Our work holds broader implications in fields such as fairness auditing~\cite{hamman2023can}, explainability~\cite{zhou2020towards,hamman2023robust,hammanjournal24}, policy regulation~\cite{WhiteHouse2022AIBill,Barocas2016BigDD}, where it can significantly contribute to the evaluation and understanding of unfairness in machine learning models. This work not only furthers the theoretical discourse but would also have significant societal implications, guiding the trajectory toward more responsible and equitable machine learning in high-stakes settings. Future work could explore efficient methods to estimate PID~\cite{venkatesh2023gaussian,goswami2023computinguniqueinformationpoisson,pakman2021estimating, kleinman2021redundant, halder2024quantifying} (also see \cite{dutta2023review} for more discussion on PID estimation). 









%

\bibliographystyle{unsrt}
\bibliography{PID}

\begin{thebibliography}{10}

\bibitem{WhiteHouse2022AIBill}
{The White House}.
\newblock Blueprint for an ai bill of rights: Making automated systems work for the american people.
\newblock \url{https://www.whitehouse.gov/ostp/ai-bill-of-rights/}, 2022.
\newblock Accessed: [30 Jan, 2024].

\bibitem{newFairSurvey}
Dana Pessach and Erez Shmueli.
\newblock A review on fairness in machine learning.
\newblock {\em ACM Comput. Surv.}, 55(3), feb 2022.

\bibitem{mehrabi2019survey}
Ninareh Mehrabi, Fred Morstatter, Nripsuta Saxena, Kristina Lerman, and Aram Galstyan.
\newblock A survey on bias and fairness in machine learning.
\newblock {\em ACM Computing Surveys (CSUR)}, 54(6):1--35, 2021.

\bibitem{barocas-hardt-narayanan}
Solon Barocas, Moritz Hardt, and Arvind Narayanan.
\newblock Fairness and machine learning, 2019.

\bibitem{varshney2021trustworthy}
Kush~R Varshney.
\newblock {\em Trustworthy Machine Learning}.
\newblock Independently Published, Chappaqua, NY, 2021.

\bibitem{kamishima2011fairness}
Toshihiro Kamishima, Shotaro Akaho, and Jun Sakuma.
\newblock Fairness-aware learning through regularization approach.
\newblock In {\em 2011 IEEE 11th International Conference on Data Mining Workshops}, pages 643--650. IEEE, 2011.

\bibitem{dutta2021fairness}
Sanghamitra Dutta, Praveen Venkatesh, Piotr Mardziel, Anupam Datta, and Pulkit Grover.
\newblock Fairness under feature exemptions: Counterfactual and observational measures.
\newblock {\em IEEE Transactions on Information Theory}, 67(10):6675--6710, 2021.

\bibitem{hamman2023can}
Faisal Hamman, Jiahao Chen, and Sanghamitra Dutta.
\newblock {Can Querying for Bias Leak Protected Attributes? Achieving Privacy With Smooth Sensitivity}.
\newblock In {\em Proceedings of the 2023 ACM Conference on Fairness, Accountability, and Transparency}, pages 1358--1368, 2023.

\bibitem{anthis2024causal}
Jacy Anthis and Victor Veitch.
\newblock Causal context connects counterfactual fairness to robust prediction and group fairness.
\newblock {\em Advances in Neural Information Processing Systems}, 36, 2024.

\bibitem{hardt2016equality}
Moritz Hardt, Eric Price, and Nati Srebro.
\newblock Equality of opportunity in supervised learning.
\newblock {\em Advances in neural information processing systems}, 2016.

\bibitem{washington2018argue}
Anne~L Washington.
\newblock How to argue with an algorithm: Lessons from the compas-propublica debate.
\newblock {\em Colo. Tech. LJ}, 17:131, 2018.

\bibitem{kleinberg2016inherent}
Jon Kleinberg, Sendhil Mullainathan, and Manish Raghavan.
\newblock Inherent trade-offs in the fair determination of risk scores.
\newblock {\em arXiv preprint arXiv:1609.05807}, 2016.

\bibitem{chouldechova2017fair}
Alexandra Chouldechova.
\newblock Fair prediction with disparate impact: A study of bias in recidivism prediction instruments.
\newblock {\em Big data}, 5(2):153--163, 2017.

\bibitem{Barocas2016BigDD}
Solon Barocas and Andrew~D. Selbst.
\newblock Big data's disparate impact.
\newblock {\em California Law Review}, 104:671, 2016.

\bibitem{dutta2020chernoff}
S.~Dutta, D.~Wei, H.~Yueksel, P.~Y. Chen, S.~Liu, and K.~R. Varshney.
\newblock Is there a trade-off between fairness and accuracy? a perspective using mismatched hypothesis testing.
\newblock In {\em International Conference on Machine Learning (ICML)}, pages 2803--2813, 2020.

\bibitem{kim2020fact}
Joon~Sik Kim, Jiahao Chen, and Ameet Talwalkar.
\newblock Fact: A diagnostic for group fairness trade-offs.
\newblock In {\em International Conference on Machine Learning}, pages 5264--5274. PMLR, 2020.

\bibitem{bertschinger_QUI}
Nils Bertschinger, Johannes Rauh, Eckehard Olbrich, J{\"u}rgen Jost, and Nihat Ay.
\newblock Quantifying unique information.
\newblock {\em Entropy}, 16(4):2161--2183, 2014.

\bibitem{ghassami2018fairness}
AmirEmad Ghassami, Sajad Khodadadian, and Negar Kiyavash.
\newblock Fairness in supervised learning: An information theoretic approach.
\newblock In {\em 2018 IEEE international symposium on information theory (ISIT)}, pages 176--180, 2018.

\bibitem{Adult}
Dheeru Dua and Casey Graff.
\newblock {UCI} machine learning repository, 2017.

\bibitem{wang2023aleatoric}
Hao Wang, Luxi He, Rui Gao, and Flavio Calmon.
\newblock Aleatoric and epistemic discrimination: Fundamental limits of fairness interventions.
\newblock In {\em Thirty-seventh Conference on Neural Information Processing Systems}, 2023.

\bibitem{zhao2022inherent}
Han Zhao and Geoffrey~J Gordon.
\newblock Inherent tradeoffs in learning fair representations.
\newblock {\em The Journal of Machine Learning Research}, 23(1):2527--2552, 2022.

\bibitem{chen2018my}
Irene Chen, Fredrik~D Johansson, and David Sontag.
\newblock Why is my classifier discriminatory?
\newblock {\em Advances in neural information processing systems}, 31, 2018.

\bibitem{long2023individual}
Carol~Xuan Long, Hsiang Hsu, Wael Alghamdi, and Flavio Calmon.
\newblock Individual arbitrariness and group fairness.
\newblock In {\em Thirty-seventh Conference on Neural Information Processing Systems}, 2023.

\bibitem{zhong2024intrinsic}
Meiyu Zhong and Ravi Tandon.
\newblock Intrinsic fairness-accuracy tradeoffs under equalized odds.
\newblock {\em arXiv preprint arXiv:2405.07393}, 2024.

\bibitem{hamman2023demystifying}
Faisal Hamman and Sanghamitra Dutta.
\newblock Demystifying local and global fairness trade-offs in federated learning using partial information decomposition.
\newblock {\em International Conference on Learning Representations (ICLR)}, 2024.

\bibitem{venkatesh2021can}
Praveen Venkatesh, Sanghamitra Dutta, Neil Mehta, and Pulkit Grover.
\newblock Can information flows suggest targets for interventions in neural circuits?
\newblock {\em Advances in Neural Information Processing Systems}, 34:3149--3162, 2021.

\bibitem{sabato2020bounding}
Sivan Sabato and Elad Yom-Tov.
\newblock Bounding the fairness and accuracy of classifiers from population statistics.
\newblock In {\em International conference on machine learning}, pages 8316--8325. PMLR, 2020.

\bibitem{menon2017cost}
Aditya~Krishna Menon and Robert~C Williamson.
\newblock The cost of fairness in classification.
\newblock {\em arXiv preprint arXiv:1705.09055}, 2017.

\bibitem{liu2022accuracy}
Suyun Liu and Luis~Nunes Vicente.
\newblock Accuracy and fairness trade-offs in machine learning: A stochastic multi-objective approach.
\newblock {\em Computational Management Science}, 19(3):513--537, 2022.

\bibitem{hertweck2022gradual}
Corinna Hertweck and Tim R{\"a}z.
\newblock Gradual (in) compatibility of fairness criteria.
\newblock In {\em Proceedings of the AAAI Conference on Artificial Intelligence}, volume~36, pages 11926--11934, 2022.

\bibitem{hsu2022pushing}
Brian Hsu, Rahul Mazumder, Preetam Nandy, and Kinjal Basu.
\newblock Pushing the limits of fairness impossibility: Who's the fairest of them all?
\newblock {\em Advances in Neural Information Processing Systems}, 35:32749--32761, 2022.

\bibitem{NIPS2016_9d268236}
Moritz Hardt, Eric Price, Eric Price, and Nati Srebro.
\newblock Equality of opportunity in supervised learning.
\newblock In {\em Advances in Neural Information Processing Systems}, volume~29, pages 3315--3323, 2016.

\bibitem{pleiss2017fairness}
Geoff Pleiss, Manish Raghavan, Felix Wu, Jon Kleinberg, and Kilian~Q Weinberger.
\newblock On fairness and calibration.
\newblock {\em Advances in neural information processing systems}, 30, 2017.

\bibitem{bell2023possibility}
Andrew Bell, Lucius Bynum, Nazarii Drushchak, Tetiana Zakharchenko, Lucas Rosenblatt, and Julia Stoyanovich.
\newblock The possibility of fairness: Revisiting the impossibility theorem in practice.
\newblock In {\em Proceedings of the 2023 ACM Conference on Fairness, Accountability, and Transparency}, pages 400--422, 2023.

\bibitem{kamishima2012fairness}
Toshihiro Kamishima, Shotaro Akaho, Hideki Asoh, and Jun Sakuma.
\newblock Fairness-aware classifier with prejudice remover regularizer.
\newblock In {\em Machine Learning and Knowledge Discovery in Databases: European Conference, ECML PKDD 2012, Bristol, UK, September 24-28, 2012. Proceedings, Part II 23}, pages 35--50. Springer, 2012.

\bibitem{calmon2017optimized}
Flavio Calmon, Dennis Wei, Bhanukiran Vinzamuri, Karthikeyan Natesan~Ramamurthy, and Kush~R Varshney.
\newblock Optimized pre-processing for discrimination prevention.
\newblock {\em Advances in neural information processing systems}, 30, 2017.

\bibitem{dutta2020information}
{S. Dutta}, Praveen Venkatesh, Piotr Mardziel, Anupam Datta, and Pulkit Grover.
\newblock An information-theoretic quantification of discrimination with exempt features.
\newblock In {\em AAAI Conference on Artificial Intelligence}, 2020.

\bibitem{cho2020fair}
Jaewoong Cho, Gyeongjo Hwang, and Changho Suh.
\newblock A fair classifier using mutual information.
\newblock In {\em 2020 IEEE international symposium on information theory (ISIT)}, pages 2521--2526. IEEE, 2020.

\bibitem{baharlouei2019r}
Sina Baharlouei, Maher Nouiehed, Ahmad Beirami, and Meisam Razaviyayn.
\newblock Renyi fair inference.
\newblock {\em arXiv preprint arXiv:1906.12005}, 2019.

\bibitem{grari2019fairness}
Vincent Grari, Boris Ruf, Sylvain Lamprier, and Marcin Detyniecki.
\newblock Fairness-aware neural reyni minimization for continuous features.
\newblock {\em arXiv preprint arXiv:1911.04929}, 2019.

\bibitem{10.1109/ISIT45174.2021.9517723}
Hao Wang, Hsiang Hsu, Mario Diaz, and Flavio~P. Calmon.
\newblock The impact of split classifiers on group fairness.
\newblock In {\em 2021 IEEE International Symposium on Information Theory (ISIT)}, page 3179–3184. IEEE Press, 2021.

\bibitem{galhotra2022causal}
Sainyam Galhotra, Karthikeyan Shanmugam, Prasanna Sattigeri, and Kush~R Varshney.
\newblock Causal feature selection for algorithmic fairness.
\newblock In {\em Proceedings of the 2022 International Conference on Management of Data}, pages 276--285, 2022.

\bibitem{NEURIPS2022_fd5013ea}
Wael Alghamdi, Hsiang Hsu, Haewon Jeong, Hao Wang, Peter Michalak, Shahab Asoodeh, and Flavio Calmon.
\newblock Beyond adult and compas: Fair multi-class prediction via information projection.
\newblock In {\em Advances in Neural Information Processing Systems}, volume~35, pages 38747--38760. Curran Associates, Inc., 2022.

\bibitem{kairouz2022generating}
Peter Kairouz, Jiachun Liao, Chong Huang, Maunil Vyas, Monica Welfert, and Lalitha Sankar.
\newblock Generating fair universal representations using adversarial models.
\newblock {\em IEEE Transactions on Information Forensics and Security}, 17:1970--1985, 2022.

\bibitem{dutta2023review}
Sanghamitra Dutta and Faisal Hamman.
\newblock A review of partial information decomposition in algorithmic fairness and explainability.
\newblock {\em Entropy}, 25(5):795, 2023.

\bibitem{tax2017partial}
Tycho~MS Tax, Pedro~AM Mediano, and Murray Shanahan.
\newblock The partial information decomposition of generative neural network models.
\newblock {\em Entropy}, 19(9):474, 2017.

\bibitem{liang2023quantifying}
Paul~Pu Liang, Yun Cheng, Xiang Fan, Chun~Kai Ling, Suzanne Nie, Richard Chen, Zihao Deng, Faisal Mahmood, Ruslan Salakhutdinov, and Louis-Philippe Morency.
\newblock Quantifying \& modeling feature interactions: An information decomposition framework.
\newblock {\em arXiv preprint arXiv:2302.12247}, 2023.

\bibitem{mohamadi2023more}
Salman Mohamadi, Gianfranco Doretto, and Donald~A Adjeroh.
\newblock More synergy, less redundancy: Exploiting joint mutual information for self-supervised learning.
\newblock {\em arXiv preprint arXiv:2307.00651}, 2023.

\bibitem{pakman2021estimating}
Ari Pakman, Amin Nejatbakhsh, Dar Gilboa, Abdullah Makkeh, Luca Mazzucato, Michael Wibral, and Elad Schneidman.
\newblock Estimating the unique information of continuous variables.
\newblock {\em Advances in neural information processing systems}, 34:20295--20307, 2021.

\bibitem{halder2024quantifying}
Barproda Halder, Faisal Hamman, Pasan Dissanayake, Qiuyi Zhang, Ilia Sucholutsky, and Sanghamitra Dutta.
\newblock Quantifying spuriousness of biased datasets using partial information decomposition.
\newblock {\em arXiv preprint arXiv:2407.00482}, 2024.

\bibitem{venkatesh2023capturing}
Praveen Venkatesh, Keerthana Gurushankar, and Gabriel Schamberg.
\newblock Capturing and interpreting unique information.
\newblock In {\em IEEE International Symposium on Information Theory (ISIT)}, pages 2631--2636, 2023.

\bibitem{blackwell1953equivalent}
David Blackwell.
\newblock Equivalent comparisons of experiments.
\newblock {\em The annals of mathematical statistics}, pages 265--272, 1953.

\bibitem{banerjee2018unique}
Pradeep~Kr Banerjee, Eckehard Olbrich, J{\"u}rgen Jost, and Johannes Rauh.
\newblock Unique informations and deficiencies.
\newblock In {\em Annual Allerton Conference on Communication, Control, and Computing (Allerton)}, pages 32--38, 2018.

\bibitem{shi2021survey}
Yuxin Shi, Han Yu, and Cyril Leung.
\newblock A survey of fairness-aware federated learning.
\newblock {\em arXiv preprint arXiv:2111.01872}, 2021.

\bibitem{dit}
R.~G. James, C.~J. Ellison, and J.~P. Crutchfield.
\newblock {dit}: a {P}ython package for discrete information theory.
\newblock {\em The Journal of Open Source Software}, 3(25):738, 2018.

\bibitem{fairtorch}
Masashi~Sode Akihiko~Fukuchi, Yoko~Yabe.
\newblock Fairtorch.
\newblock \url{https://github.com/wbawakate/fairtorch}, 2021.

\bibitem{zhou2020towards}
Jianlong Zhou, Fang Chen, and Andreas Holzinger.
\newblock Towards explainability for ai fairness.
\newblock In {\em International Workshop on Extending Explainable AI Beyond Deep Models and Classifiers}, pages 375--386. Springer, 2020.

\bibitem{hamman2023robust}
Faisal Hamman, Erfaun Noorani, Saumitra Mishra, Daniele Magazzeni, and Sanghamitra Dutta.
\newblock Robust counterfactual explanations for neural networks with probabilistic guarantees.
\newblock In {\em International Conference on Machine Learning}, pages 12351--12367. PMLR, 2023.

\bibitem{hammanjournal24}
Faisal Hamman, Erfaun Noorani, Saumitra Mishra, Daniele Magazzeni, and Sanghamitra Dutta.
\newblock Robust algorithmic recourse under model multiplicity with probabilistic guarantees.
\newblock {\em IEEE Journal on Selected Areas in Information Theory}, pages 1--1, 2024.

\bibitem{venkatesh2023gaussian}
Praveen Venkatesh, Corbett Bennett, Sam Gale, Tamina~K. Ramirez, Greggory Heller, Severine Durand, Shawn~R Olsen, and Stefan Mihalas.
\newblock Gaussian partial information decomposition: Bias correction and application to high-dimensional data.
\newblock In {\em Thirty-seventh Conference on Neural Information Processing Systems}, 2023.

\bibitem{goswami2023computinguniqueinformationpoisson}
Chaitanya Goswami, Amanda Merkley, and Pulkit Grover.
\newblock Computing unique information for poisson and multinomial systems, 2023.

\bibitem{kleinman2021redundant}
Michael Kleinman, Alessandro Achille, Stefano Soatto, and Jonathan~C Kao.
\newblock Redundant information neural estimation.
\newblock {\em Entropy}, 23(7):922, 2021.

\end{thebibliography}

\appendix
\section{Proofs for Theorems in Section~\ref{sec:fairdecomp}}
\revisitingimposibility*
\begin{proof}

For completeness, we first show the non-negativity of PID terms for the PID definition that we are using in this work (also see \cite{bertschinger_QUI}):

$\uni{Z}{\hat{Y}|Y} = \min_{Q \in \Delta_p} \mutd{Q}{Z}{\hat{Y}|Y}$ is non-negative since the conditional mutual information is non-negative by definition.

Similarly argument holds for $\uni{Z}{Y|\hat{Y}}$.

$\syn{Z}{\hat{Y},Y} = \mut{Z}{\hat{Y}|Y} - \min_{Q \in \Delta_p} \mutd{Q}{Z}{\hat{Y}|Y} 
    \geq \mut{Z}{\hat{Y}|Y} - \mut{Z}{\hat{Y}|Y} 
    = 0.$

The  Redundant information: 
\begin{align*}
    \rd{Z}{\hat{Y},Y} &= \mut{Z}{\hat{Y}}- \min _{Q \in \Delta_p} \mathrm{I}_Q(Z ; \hat{Y} | S) = \max_{Q \in \Delta_p} \mutd{Q}{\hat{Y}}{Z} - \mutd{Q}{Z} {\hat{Y}|Y}
\end{align*}
First equality holds by definition of redundant information. Second equality holds  since   marginals on $(\hat{Y},Z)$ is fixed in $\Delta_p$, hence, $\max_{Q \in \Delta_p} \mutd{Q}{\hat{Y}}{Z} = \mut{\hat{Y}}{Z} $.

To prove non-negativity of redundant information, we construct a distribution $Q_0$ such that:
\begin{align*}
    \Pr_{Q_0}(Z=z,\hat{Y}=y,Y=y) = \frac{\Pr(Z=z,\hat{Y}=y)\Pr(Z=z,Y=y)}{\Pr(Z=z)} 
\end{align*}

Next, we show $Q_0 \in \Delta_p$. Recall the set \(\Delta_p\) in Definition \ref{def:bert_def}:
\begin{align*}
    \Delta_p = \{Q \in \Delta: \Pr_{Q}(Z=z,\hat{Y}=y) = \Pr(Z=z, \hat{Y}=y), \Pr_{Q}(Z=z, Y=y) = \Pr(Z=z, Y=y)\}.
\end{align*}

\begin{align*}
    \Pr_{Q_0}(Z=z,\hat{Y}=y) &=  \sum_y \Pr_{Q_0}(Z=z,\hat{Y}=y,Y=y)=  \sum_y \frac{\Pr(Z=z,\hat{Y}=y)}{\Pr(Z=z)}\Pr(Z=z,Y=y) \\
    &= \frac{\Pr(Z=z,\hat{Y}=y)}{\Pr(Z=z)}   \sum_y \Pr(Z=z,Y=y) = \Pr(Z=z,\hat{Y}=y).
\end{align*}
\begin{align*}
    \Pr_{Q_0}(Z=z,Y=y) &=  \sum_{\hat{y}} \Pr_{Q_0}(Z=z,\hat{Y}=y,Y=y)=  \sum_{\hat{y}} \frac{\Pr(Z=z,\hat{Y}=y)\Pr(Z=z,Y=y)}{\Pr(Z=z)} \\
    &= \frac{\Pr(Z=z,Y=y)}{\Pr(Z=z)}   \sum_{\hat{y}} \Pr(Z=z,\hat{Y}=y) = \Pr(Z=z,Y=y). 
\end{align*}

Marginals of $Q_0$ satisfy conditions on set $\Delta_p$, hence $Q_0\in \Delta_p$. Also, note that by construction of $Q_0$, $\hat{Y}$ and $Y$ are independent conditioned on $Z$, i.e., $\mutd{Q_0}{\hat{Y}}{Y|Z} = 0$. Hence, we have:
\begin{align*}
    \rd{Z}{\hat{Y},Y} & \stackrel{(a)}{=} \max_{Q \in \Delta_p} \mutd{Q}{Z}{\hat{Y}} - \mutd{Q}{Z}{\hat{Y}|Y} \\
    & \stackrel{(b)}{\geq} \mutd{Q_0}{Z}{\hat{Y}} - \mutd{Q_0}{Z}{\hat{Y}|Y}\\
    & \stackrel{(c)}{=}  H_{Q_0}(Z)+H_{Q_0}(\hat{Y}) - H_{Q_0}(Z,\hat{Y}) - H_{Q_0}(Z|Y) - H_{Q_0}(\hat{Y}|Y)+H_{Q_0}(Z,\hat{Y}|Y)\\
    & \stackrel{(d)}{=} \mutd{Q_0}{\hat{Y}}{Y} - \mutd{Q_0}{\hat{Y}}{Y|Z} \\
    &\stackrel{(e)}{=} \mutd{Q_0}{\hat{Y}}{Y} \stackrel{(f)}{\geq} 0.
\end{align*}
Here, (\textit{a}) hold from definition of $\rd{Z}{\hat{Y},Y}$, (\textit{b}) hold since $Q_0 \in \Delta_p$, (\textit{c})-(\textit{d}) holds from expressing mutual information in terms of entropy, (\textit{e}) hold since $\mutd{Q_0}{\hat{Y}}{S|Z} = 0$, (\textit{f}) holds from non-negativity property of mutual information.


Since mutual information $\mut{Z}{\Yh,Y} > 0$, we can use the PID framework to decompose this mutual information as:
\begin{align*}
\mut{Z}{\Yh,Y} = \uniZ + \uniZY + \redZ + \synZ.
\end{align*}

Since each of the PID components $\uniZ$, $\uniZY$, $\redZ$, and $\synZ$ is nonnegative, the condition $\mut{Z}{\Yh,Y} > 0$ implies that at least one of these terms must be strictly positive.

Next, according to Proposition~\ref{prop:decomposition}:
\begin{align*}
\SP & = \uni{Z}{\Yh|Y} + \rd{Z}{\Yh, Y}, \\
\EO & = \uni{Z}{\Yh|Y} + \syn{Z}{\Yh,Y}, \\
\PP & = \uni{Z}{Y|\Yh} + \syn{Z}{\Yh,Y}.
\end{align*}

Given that $\mut{Z}{\Yh,Y} > 0$ ensures that at least one of $\uniZ$, $\redZ$, $\synZ$, or $\uniZY$ is nonnegative, it follows that:
- If $\uni{Z}{\Yh|Y} > 0$, then $\SP > 0$.
- If $\syn{Z}{\Yh,Y} > 0$, then $\EO > 0$ or $\PP > 0$.
- If $\rd{Z}{\Yh, Y} > 0$, then $\SP > 0$.
- If $\uni{Z}{Y|\Yh} > 0$, then $\PP > 0$.

Therefore, the presence of positive mutual information $\mut{Z}{\Yh,Y}$ guarantees that at least one of the fairness measures ($\SP$, $\EO$, $\PP$) will be nonzero.

Conversely, if all these unfairness measures are zero, then by the definitions given in Proposition~\ref{prop:decomposition} and non-negativity of PID terms, all the PID terms must be zero, which implies $\mut{Z}{\Yh,Y} = 0$.
\end{proof}

\DatasetDependentRelationships*

\begin{proof}
We begin by expressing the mutual information between $Z$ and $Y$ using PID:
\begin{equation}
\mut{Z}{Y} = \uni{Z}{Y|\Yh} + \rd{Z}{Y,\Yh}.
\end{equation}

Next, we examine the contributions of these PID terms to the fairness measures: The term $\uni{Z}{Y|\Yh}$ contributes to the Predictive Parity Gap ($\PP$).  The term $\rd{Z}{Y,\Yh}$ contributes to the Statistical Parity Gap ($\SP$).

Given that $\mut{Z}{Y} > 0$, it follows that the sum of $\uni{Z}{Y|\Yh}$ and $\rd{Z}{Y,\Yh}$ is positive, i.e.,
$\mut{Z}{Y} = \uni{Z}{Y|\Yh} + \rd{Z}{Y,\Yh} > 0.$

Since both $\uni{Z}{Y|\Yh}$ and $\rd{Z}{Y,\Yh}$ are nonnegative, the fact that their sum is positive implies that at least one of them must be strictly positive. Thus, we have two cases to consider:

1. If $\uni{Z}{Y|\Yh} > 0$, then the Predictive Parity Gap ($\PP$) must be greater than zero.

2. If $\rd{Z}{Y,\Yh} > 0$, then the Statistical Parity Gap ($\SP$) must be greater than zero.

Therefore, if $\mut{Z}{Y} > 0$, it is guaranteed that either $\SP > 0$ or $\PP > 0$.
\end{proof}

\section{Proofs for Theorems in Section~\ref{sec:tradeoff}}
\tradeoffone*

\begin{proof}
Given that Statistical Parity is zero (\(\SP = 0\)), we have:
\begin{align}
\SP &= \uni{Z}{\Yh|Y} + \rd{Z}{\Yh, Y} = 0.
\end{align}
Since all PID terms are non-negative, this implies:
\begin{align}
\uni{Z}{\Yh|Y} = 0 \quad \text{and} \quad \rd{Z}{\Yh, Y} = 0.
\end{align}

The Equalized Odds gap (\(\EO\)) and Predictive Parity gap (\(\PP\)) can be expressed as:
\begin{align}
\EO &= \uni{Z}{\Yh|Y} + \syn{Z}{\Yh,Y} = \syn{Z}{\Yh,Y}, \\
\PP &= \uni{Z}{Y|\Yh} + \syn{Z}{\Yh,Y}.
\end{align}
Since \(\uni{Z}{\Yh|Y} = 0\), we have:
\begin{align}
\PP = \uni{Z}{Y|\Yh} + \syn{Z}{\Yh,Y} \geq \syn{Z}{\Yh,Y} = \EO.
\end{align}

Furthermore, if \(\mut{Z}{Y} = 0\), we have:
\begin{align}
\mut{Z}{Y} = \uni{Z}{Y|\Yh} + \rd{Z}{\Yh, Y} = 0.
\end{align}
This implies:
\begin{align}
\uni{Z}{Y|\Yh} = 0 \quad \text{and} \quad \rd{Z}{\Yh, Y} = 0.
\end{align}
Therefore:
\begin{align}
\PP = \uni{Z}{Y|\Yh} + \syn{Z}{\Yh,Y} = \syn{Z}{\Yh,Y} = \EO.
\end{align}

Thus, we have: \(\PP \geq \EO\) when \(\SP = 0\), and \(\PP = \EO\) when \(\mut{Z}{Y} = 0\). 
\end{proof}

\tradeofftwo*

\begin{proof}
Given that Predictive Parity is satisfied (\(\PP = 0\)), we have:
\begin{align}
\PP = \uni{Z}{Y|\Yh} + \syn{Z}{\Yh,Y} = 0.
\end{align}
Since all PID terms are non-negative, this implies:
\begin{align}
\uni{Z}{Y|\Yh} = 0 \quad \text{and} \quad \syn{Z}{\Yh,Y} = 0.
\end{align}

The Statistical Parity gap (\(\SP\)) and Equalized Odds gap (\(\EO\)) can be expressed as:
\begin{align}
\SP &= \uni{Z}{\Yh|Y} + \rd{Z}{\Yh, Y}, \\
\EO &= \uni{Z}{\Yh|Y} + \syn{Z}{\Yh,Y} = \uni{Z}{\Yh|Y}.
\end{align}

Since \(\uni{Z}{Y|\Yh} = 0\), we have:
\begin{align}
\SP = \uni{Z}{\Yh|Y} + \redZ \geq \uni{Z}{\Yh|Y} = \EO.
\end{align}

Thus, \(\SP \geq \EO\) when \(\PP = 0\).

Furthermore, if \(\mut{Z}{Y} = 0\), we have:
\begin{align}
\mut{Z}{Y} = \uni{Z}{Y|\Yh} + \rd{Z}{\Yh, Y} = 0.
\end{align}
This implies:
\begin{align}
\uni{Z}{Y|\Yh} = 0 \quad \text{and} \quad \rd{Z}{\Yh, Y} = 0.
\end{align}
Therefore: $\SP = \uniZ \quad \text{and} \quad \EO = \uniZ.$

Hence, when \(\mut{Z}{Y} = 0\), and $\PP = 0$, we have: \(\SP = \EO\).
\end{proof}

\tradeoffthree*

\begin{proof}
Given that Equalized Odds is satisfied (\(\EO = 0\)), we have:
\begin{align}
\EO = \uni{Z}{\Yh|Y} + \syn{Z}{\Yh,Y} = 0.
\end{align}
Since all PID terms are non-negative, it follows that:
\begin{align}
\uni{Z}{\Yh|Y} = 0 \quad \text{and} \quad \syn{Z}{\Yh,Y} = 0.
\end{align}

The Statistical Parity gap (\(\SP\)) and Predictive Parity gap (\(\PP\)) can now be expressed as:
\begin{align}
\SP &= \uni{Z}{\Yh|Y} + \rd{Z}{\Yh, Y} = \rd{Z}{\Yh, Y}, \\
\PP &= \uni{Z}{Y|\Yh} + \syn{Z}{\Yh,Y} = \uni{Z}{Y|\Yh}.
\end{align}

Given that \(\mut{Z}{Y} > 0\), we can write the mutual information between \(Z\) and \(Y\) as:
\begin{align}
\mut{Z}{Y} = \uni{Z}{Y|\Yh} + \rd{Z}{\Yh, Y}.
\end{align}

Substituting \(\PP\) and \(\SP\) into this equation, we get:
\begin{align}
\mut{Z}{Y} = \PP + \SP.
\end{align}

Since \(\mut{Z}{Y}\) is fixed for a given dataset, this equation demonstrates that \(\SP\) and \(\PP\) have an inverse relationship:
\begin{align}
\SP = \mut{Z}{Y} - \PP.
\end{align}

Therefore, increasing the Statistical Parity gap (\(\SP\)) will lead to a decrease in the Predictive Parity gap (\(\PP\)), and vice versa.
\end{proof}









\end{document}